\documentclass[draftcls,peerreview,12pt,onecolumn]{IEEEtran}
\usepackage{multirow}
\usepackage{amsmath}

\usepackage{amssymb}
\usepackage{graphicx}
\usepackage{cite}
\usepackage{cite}
\usepackage{flushend}
\usepackage{amsthm}
\usepackage{booktabs}
\usepackage{epstopdf}
\usepackage{booktabs}
\usepackage{balance}
\usepackage[outercaption]{sidecap}

\newtheorem{proposition}{Proposition}
\newtheorem{corollary}{Corollary}
\newtheorem*{lemma*}{Lemma}

\usepackage{xcolor}
\usepackage{soul}

\title{A Unified Effective Capacity Performance Analysis of $L_p$-norm Diversity Reception over Arbitrary and Correlated Generalized Fading Channels*}
\author{K.~Denia~Kanellopoulou, ~\IEEEmembership{Student Member,~IEEE}, Kostas~P.~Peppas,~\IEEEmembership{Senior Member,~IEEE} and P.~Takis~Mathiopoulos,~\IEEEmembership{Senior Member,~IEEE}

\thanks{\hl{*This manuscript was submitted on Sept. 30, 2017, for possible publication in the IEEE TCOM as TCOM-TPS-17-1021. On Dec. 17, 2017, it was
rejected with a possibility of resubmission by the reviewing editor Prof. Sami Muhaidat.  Motivated by the reviewers' comments, we have 
significantly extended the research reported in TCOM-TPS-17-1021 and on March 4, 2019, we have resubmitted the revised manuscript for possible 
publication as TCOM-TPS-19-0147 (on arXiv it was uploaded on Feb. 2, 2019 as arXiv:1703.06941v2 [cs.IT]). Since the originally submitted manuscript
(i.e. TCOM-TPS-17-1021) was uploaded to arXiv on Oct. 10, 2019, i.e. long after its original submission (on Sept. 30, 2017). It must be underlined that, 
although it is referenced as arXiv:1703.06941v3[cs.IT], it reports the research which we have carried out long before the submission of its revised version 
TCOM-TPS-19-0147 (March 4, 2019) referenced as arXiv:1703.06941v2[cs.IT] (Feb. 2, 2019)}}

\thanks{K.~Denia~Kanellopoulou and P.~Takis~Mathiopoulos are with the Department of Informatics and Telecommunications, National and Kapodistrian University of Athens, 15784 Ilisia, Athens, Greece (e-mail: mathio@di.uoa.gr).}
\thanks{Kostas~P.~Peppas is with the Department of Informatics and Telecommunications, University of Peloponnese, 22100 Tripoli, Greece. (e-mail: peppas@uop.gr)}
}
\begin{document}
\maketitle
\vspace{0.5cm}
\begin{abstract}
The effective capacity (EC) has been recently established as a rigorous alternative to the classical Shannon's
ergodic capacity since it accounts for the delay constraints imposed by future wireless
applications and their impact on the overall system performance.
This paper presents a novel moment generating function (MGF)-based framework for the unified EC performance analysis of a generic $L_p$-norm diversity combining scheme operating over arbitrary and correlated generalized
fading channels and a maximum delay constraint.
The $L_p$-norm diversity is a generic diversity structure which includes as special cases various well-known diversity schemes such as equal gain combining (EGC) and maximal ratio combining (MRC).
For MRC, the proposed methodology reduces to a previously published MGF-based approach for the evaluation of the EC, whereas, for EGC, analytical approach presented is novel and the associated performance evaluation results have not been published previously in the open technical literature.
Based on this methodology, novel analytical closed-form expressions for the EC performance of dual branch $L_p$-norm diversity receivers operating over Gamma shadowed generalized Nakagami-$m$ fading channels are deduced.
For diversity order greater than two, a novel analytical approach for the asymptotic EC performance analysis is also developed and evaluated, revealing how basic system parameters affect the overall system performance.
The overall mathematical formalism is validated with selected numerical and equivalent simulation performance evaluation results thus confirming the correctness of the proposed unified analytical methodology.
\end{abstract}
\newcounter{mytempeqncnt}
\section{Introduction}\label{Sec:Intro}

\IEEEPARstart{R}{eal}-time emerging applications such as voice over internet protocol
(IP), interactive and multimedia streaming, interactive gaming, mobile TV and computing
are largely delay-sensitive, which implies that
the data will expire if it is not successfully delivered within a certain time frame. As such, an alternative quality of service (QoS) metric that is able to capture the
end-to-end communication delay is required. Unfortunately, the conventional notion of Shannon or outage capacity cannot account for the delay aspect. The effective capacity (EC)
was introduced in \cite{J:Wu} as an alternative metric that quantifies the system performance under QoS limitation.


Recently, the concept of EC has attracted the attention of the wireless research community for the performance analysis of single-input single-output (SISO) and multiple-input single-output (MISO)
communication systems with transmit maximal ratio combining (MRC)\footnote{
{The SNR at the output of MISO receivers with transmit MRC is essentially the same as that obtained when MRC is employed
in single-input multiple-out (SIMO) systems (divided by a factor equal to the number of transmitting antennas) and, thus, their EC performances will be identical as well.
In fact, when full channel state information (CSI) is assumed at the transmitter, a MISO system can be viewed as
the dual of a SIMO system with MRC.}}, e.g. see \cite{J:Matthaiou, J:Zhang, C:Zhang, J:Zhong, J:Guo, J:You2, J:Li, J:Ji1,C:Ji2, J:You}.
The subject of evaluating the performance of such systems has been investigated in the past in several research works, assuming either independent \cite{J:Matthaiou, J:Zhang, C:Zhang} or correlated \cite{J:Zhong, J:Guo} fading channels. Specifically, in \cite{J:Matthaiou}, the EC of MISO systems assuming independent Nakagami-$m$, Rice and generalized-K channels was addressed. In \cite{J:Zhang} and \cite{C:Zhang}, the EC of MISO systems assuming independent $\kappa$-$\mu$ and $\eta$-$\mu$ fading channels, respectively, was investigated.
In \cite{J:Zhong}, the EC of multiple-antenna systems operating in fading channels with correlation and keyholes has been considered.
In \cite{J:Guo}, infinite series representations of the EC of MISO systems operating over correlated Rice and Nakagami-$m$ fading channels have been derived. The EC of SISO and MISO systems operating over Weibull and shadowed $\kappa-\mu$ fading channels has been addressed in \cite{J:You2}, and \cite{J:Li}, respectively.
In \cite{J:Ji1,C:Ji2, J:You}, various analytical approaches for the computation of the EC of MISO systems operating over arbitrary generalized fading channels have been presented. These approaches utilize the moment generating function (MGF) to deduce single-integral expressions for the EC. Specifically, \cite{J:Ji1} and \cite{C:Ji2} have presented an MGF-based approach for the EC evaluation of MISO systems by utilizing a Laplace transform of the MGF of the receiver's output signal-to-noise ratio (SNR). More recently, in \cite{J:You}, a unified analytical framework for the evaluation of the EC of MISO systems by employing the so-called H-transform techniques has been proposed. However, as it will be explained later on in this paper, these approaches deal with MRC receivers and their methodology is not applicable to other diversity schemes.


Another popular diversity scheme, namely equal gain combining (EGC), is regarded as a practical alternative to MRC since its performace is comparable to that of MRC but exhibits lower implementation complexity \cite[Chapter 9, p. 278]{B:Alouini}.
Nevertheless, to the best of our knowledge the EC of EGC receivers operating over arbitrary and/or correlated generalized fading channels has not been considered in the open technical literature yet.
This is mainly because of the inherent difficulty in obtaining simple, mathematically tractable expressions for the probability density function (PDF) of the sum of fading envelopes \cite[Chapter 9, p. 282]{B:Alouini}. Moreover, various performance evaluation approaches presented in the past, i.e. \cite{J:Matthaiou, J:Zhang, C:Zhang, J:Guo, J:Li, J:You2, J:Ji1,C:Ji2, J:You}, are valid only for MRC diversity receivers and thus can't be applied
to the evaluation of the EC performance of other diversity receivers, such as the EGC.

Motivated by the above, in this paper
we follow a more general approach by considering a signal combining structure whose output is
the $L_p$-norm of the fading envelopes of each of the $L$ diversity branches and analyzing its EC performance in a unified way, by developing a new MGF-based methodology.
As it will be explained in Section~\ref{Sec:System}, where the system model will be presented, this is a generic diversity structure which includes several well known diversity schemes, including EGC ($L_1$-norm diversity) and MRC ($L_2$-norm diversity), as special cases.
Note that, although for both of these well-known diversity schemes various EC performance evaluation results will be obtained, our EC study for the EGC receiver is, to the best of our knowledge, for the first time published in the open technical literature. Furthermore, and apart from introducing introducing this generic methodology, the main novel contributions of the paper can be
summarized as follows.


\begin{itemize}
  \item A novel, single integral expression for the unified EC performance analysis of $L_p$-norm diversity receivers is deduced. This generic expression can be used to obtain the EC performance for a variety of diversity schemes operating over arbitrary and correlated generalized fading by means of standard numerical integration techniques, provided that the MGF of the $p$-th power of the fading envelope, where $p$ is an arbitrary positive real number, can be obtained. 
    The proposed analysis includes other previously published MGF-based approaches presented in \cite{J:Ji1,C:Ji2, J:You}, where EC performance evaluation results for the MRC receiver over arbitrary and correlated generalized fading channels have been obtained, as special cases;
  \item Derivation of a novel, closed-form, analytical expression for the EC of dual-branch $L_p$-norm diversity receivers, operating over a so-called Gamma-shadowed generalized Nakagami-$m$ (GSNM) fading environment. The motivation behind the selection of the GSNM fading distribution is its versatility, as it includes as special or limiting cases the Nakagami-$m$, the generalized gamma, the generalized-K and the log-normal distribution \cite{J:YilmazEGC}.
      Moreover, it can accurately model wireless propagation in high frequency (60 GHz and above) communication systems \cite{J:YilmazEGC}.
   For the special cases of generalized-K and Nakagami-$m$ fading, closed-form analytical expressions for the EC of dual branch EGC diversity receivers are also presented;
    \item In order to obtain further insights into the effect of system parameters, such as delay constraints, fading parameters and number of antennas,
     a novel asymptotic analysis is proposed thus assessing the EC of MRC and EGC diversity receivers at low- and high-SNR regimes.
\end{itemize}
It is noted that the accuracy of the overall analytical methodology is substantiated with numerical results, accompanied with equivalent performance evaluation results obtained by means of Monte-Carlo simulations.

The rest of the paper is organized as follows. In Section~\ref{Sec:System}, the system and channel model are introduced. In Section \ref{Sec:Result}, an exact analytical expression for the EC using a unified MGF-based approach is developed. Analytical expressions for the EC of dual-branch $L_p$-norm diversity receivers, operating over GSGN fading channels can be found in Section \ref{Sec:Anal}, while Section \ref{Sec:Asymptotic} presents generic asymptotic expressions for EC in high- and low-SNR regimes.
Numerical results are provided in Section \ref{Sec:Results}, validating the proposed analysis. Finally, Section \ref{Sec:Conclusions} summarizes the paper and its main findings.
\emph{Notations}: ${\mathbb E}\langle\cdot\rangle$ denotes expectation, $f_X(\cdot)$ denotes the PDF of the random variable (RV) $X$,
$\mathcal{M}_X(\cdot)$ denotes the MGF of the random variable $X$, $\Delta(k, a) = \left\{\frac{a}{k},\frac{a+1}{k},\ldots, \frac{a+k-1}{k}\right\}$,
$\mathbf{I}$, $\det(\cdot)$ and $(\cdot)^{-1}$ denote the $L \times L$ identity matrix, matrix determinant, and matrix
inversion, respectively, $\parallel \cdot \parallel^2_F$ denotes the squared Frobenius norm,
$\mathbb{L}^{-1}\left\{F(s); s; t\right\}$ denotes the inverse Laplace transform of $F(s)$,
$\mathbb{N}$ is the set of positive integers, $\mathbb{R}^+$ is the set of positive reals,
$J_{a}\left(\cdot\right)$  is the Bessel function of the first kind and order $a$ \cite[Eq. (8.402)]{B:Gradshteyn_00}, $I_{a}\left(\cdot\right)$ is the modified Bessel function of the first kind and order $a$ \cite[eq. (8.431)]{B:Gradshteyn_00}, 
$\Gamma\left(\cdot\right)$ is the Gamma function \cite[Eq. (8.310/1)]{B:Gradshteyn_00},
$G \,\substack{ m , n\\ p , q} \left[\cdot\right]$ is the Meijer's G-function \cite[Eq. (9.301)]{B:Gradshteyn_00},
$G \, \substack{ m_1 , n_1 : m_2 , n_2 : m_3 , n_3 \\ p_1 , q_1 :p_2 , q_2 : p_3 , q_3} $ is the bivariate Meijer's G-function \cite{B:HFox},
$H \,\substack{ m , n\\ p , q} \left[\cdot\right]$ is the Fox's G-function, \cite{B:HFox} and
$H \, \substack{ m_1 , n_1 : m_2 , n_2 : m_3 , n_3 \\ p_1 , q_1 :p_2 , q_2 : p_3 , q_3} $ is the bivariate Fox's H-function \cite{B:HFox}.

.

\section{System and Channel Models}\label{Sec:System}

Consider an 
$L$-branch diversity receiver operating in the presence of arbitrary generalized fading and AWGN.
Let $\mathcal{R}_\ell$ be the fading envelope at the $\ell$-th branch, $\forall \ell \in \{1,2,\ldots,L\}$. The $L_p$ norm of
the random vector $\vec{{\mathbf{R}}} = \{\mathcal{R}_1, \mathcal{R}_2, \ldots , \mathcal{R}_L\}$
is defined as \cite[p. 64]{B:Rudin}
\begin{equation}\label{Eq:pnorm}
 \| \vec{{\mathbf{R}}}  \|_p \triangleq \left(\sum_{\ell = 1}^L \mathcal{R}_\ell^p \right)^{\frac{1}{p}}
\end{equation}
where $p$ is a positive real number\footnote{
In \cite{J:Nasri}, a different $L_p$-norm has been proposed for the purpose of evaluating the Euclidean distance between two points in the signal constellation diagram raised to the $p$-th power. Using this metric, asymptotic bit error rate results were presented for different diversity schemes operating in non-Gaussian and interference channels.
}. A generic analytical expression for the instantaneous SNR, $\gamma_{\rm end}$, at the output
of the $L_p$-norm diversity receiver
can be deduced as \cite[Eq. (2)]{J:YilmazEGC}
\begin{equation}\label{Eq:SNR}
\gamma_{\rm end} = \frac{E_s}{N_0\sqrt{L^{1-p+q}}} \left(\| \vec{{\mathbf{R}}}  \|_p\right) ^{pq}
\end{equation}
where $E_s/N_0$ is the SNR per symbol, with $E_s$ being the average symbol energy and $N_0$ is the single-sided power spectral density of the AWGN. Note that for $(p,q) = (1,2)$ the EGC
diversity scheme results, whereas for $(p,q) = (2,1)$ the MRC. 
As far as the queuing model is concerned, a simple first-input first-output (FIFO) buffer
with constant arrival rate (source data rate) at the
transmitter data link layer is considered. By considering ideal modulation and
coding at the source physical layer, the service rate of the buffer
will be equal to the instantaneous channel capacity. Therefore, using \cite[Eq. (4)]{J:Zhong} and assuming that the transmitter sends uncorrelated
circularly symmetric zero-mean complex Gaussian signals, the EC can be expressed as
\begin{align}\label{Eq:Defcapacity}
R(\theta) = -\frac{1}{\theta T B}\ln\left[\mathbb{E}\left\langle\left(1+\gamma_{\rm end}\right)^{-A}\right\rangle\right]
 \end{align}
where $A \triangleq \theta T B/\ln 2$ represents a metric of delay constraint, with $B$ denoting the bandwidth of the
system, $T$ the fading block length and $\theta$ the asymptotic decay-rate of the buffer occupancy.
It is noted that a small value for $\theta$ corresponds to a slow decaying rate thus having less stringent QoS requirements, while a larger one refers to faster decaying rates and thus more stringent QoS requirements.

By substituting \eqref{Eq:SNR} into \eqref{Eq:Defcapacity}, it becomes evident that the evaluation of the EC assuming arbitrarily distributed $\mathcal{R}_\ell$ with generalized correlation, involves the numerical evaluation of the following $L$-fold integral,
\begin{equation}
\mathcal{I} = \int_{\vec{\mathbf{{R}}}}\left[1+\frac{E_s}{N_0\sqrt{L^{1-p+q}}}\left(\sum_{\ell =0}^L \mathcal{R}_\ell^p\right)^q\right]^{-A}
f_{\vec{\mathbf{{R}}}}(\vec{\mathbf{R}})\mathrm{d}\vec{\mathbf{R}}
\end{equation}
where $f_{\vec{\mathbf{{R}}}}(\vec{\mathbf{R}})$ is the joint PDF of $\vec{{\mathbf{R}}} \triangleq \{\mathcal{R}_1, \mathcal{R}_2, \ldots , \mathcal{R}_L\}$. Clearly, this expression becomes prohibitively complex even for small values of $L$. In fact, even beyond only three branches, i.e. $L$ $>$ 3, such an approach becomes computationally intractable and numerical results may not even converge.
In order to solve this cumbersome statistical problem, a unified MGF-based approach for the numerical evaluation of $R(\theta)$ will be presented next that provides a generic single integral expression
for the EC of EGC and MRC diversity combiners over arbitrarily correlated generalized fading channels.

\section{The Proposed MGF-based approach}\label{Sec:Result}
\subsection{Main Result}
In order to conveniently present our generic approach, the following lemma is proved first.
\begin{lemma*}
Let $X$ be a positive RV with PDF $f_X(x)$ and MGF $\mathcal{M}_X(u)$. Let also $g(X)$ be a function of $X$ for which it is assumed that the following inverse Laplace transform $h(u) = \mathbb{L}^{-1}\left\{g(x); x; u\right\}$
exists. Then, the expectation $\mathbb{E}\langle g(X) \rangle$ can be expressed in terms of $\mathcal{M}_X(u)$ as
\begin{equation}
\mathbb{E}\langle g(X) \rangle = \int_0^{\infty} h(u) \mathcal{M}_X(u) \mathrm{d}u.
\end{equation}
 \end{lemma*}
\begin{proof}
By exploiting the definition of the MGF, i.e., $\mathcal{M}_{X} (u) \triangleq \int_0^{\infty} \exp(-u x)f_{X}(x)\mathrm{d}x$, the expectation $\mathbb{E}\langle g(X) \rangle$ can be written as
\begin{equation}
\begin{split}
\mathbb{E}\langle g(X) \rangle & = \int_0^{\infty}f_X(x)g(x)\mathrm{d}x \\
& = \int_0^{\infty}f_X(x)\left[\int_0^{\infty}\exp(-x\,u) h(u) \mathrm{d}u\right]\mathrm{d}x \\
& = \int_0^{\infty}h(u)\left[\int_0^{\infty}\exp(-x\,u) f_X(x) \mathrm{d}x\right]\mathrm{d}u \\
& = \int_0^{\infty}h(u)\mathcal{M}_{X} (u)\mathrm{d}u,
\end{split}
\end{equation}
which completes the proof.
\end{proof}
Based on the above lemma, a unified MGF-based approach for the evaluation of the EC of diversity receivers can be deduced as follows.
\begin{proposition}\label{Prop:Proposition1}
The EC of $L$-branch diversity receivers over arbitrary
not necessarily independent nor identically distributed generalized fading channels can be expressed in terms of a single integral as
\begin{equation}\label{Eq:ERMGFbased}
R(\theta) = -\frac{1}{\theta T B}\ln{\left[\frac{1}{\Gamma(A)}\int_0^{\infty} C_{q}(u)\mathcal{M}_{\vec{{\mathbf{R}}}^p}(K_{p,q}u)\mathrm{d}u  \right]}
\end{equation}
where $\mathcal{M}_{\vec{{\mathbf{R}}}^p}(u) \triangleq \mathbb{E}\langle\exp(-u\sum_{\ell=1}^L \mathcal{R}^p )\rangle$ is the joint MGF of
the $p$-th exponent of the random vector $\vec{{\mathbf{R}}}$, $K_{p,q} = \sqrt[q]{E_s/N_0 L^{(p-q-1)/2}}$ and the function $C_{q}(u)$ is given by
\begin{equation}\label{Eq:aux}
C_{q}(u) = \begin{cases}
    (2\pi)^{\frac{q-1}{2}}\frac{\sqrt{q}}{u}G\substack{1,1\\q+1,1} \left[ \frac{1}{u^q}  \left| \substack{ 1-A, \Delta(q,0) \\ 0 } \right.\right], &  q \in \mathbb{N} \\
    \\
    \frac{u^{-1}}{q}H\substack{1,1\\1,2} \left[ u  \left| \substack{ (1,1/q) \\ (A,1/q), (1,1) } \right.\right] & q \in \mathbb{R}^+.
  \end{cases}
\end{equation}
\end{proposition}
\begin{proof}
Let us define the RV $X \triangleq \sum_{\ell=1}^L \mathcal{R}^p$ and the auxiliary function $C_{q}(u) = \mathbb{L}^{-1}\left\{\left(1+x^q\right)^{-A}; x; u\right\}$. Then, by
employing the previous lemma,
the expectation of $\left(1+X^q\right)^{-A}$ can be expressed as
\begin{align}
 \mathbb{E}\left\langle\left(1+X^q\right)^{-A}\right\rangle = \int_0^{\infty}C_{q}(u)\mathcal{M}_X(u)\mathrm{d}u.\nonumber
\end{align}
By considering the identity \cite[Eq. (8.4.2.5)]{B:Prudnikov3}
\begin{equation}
\begin{split}
(1+x^q)^{-A} & = \frac{1}{\Gamma(A)}G\substack{1,1\\1,1} \left[ x^{q}  \left| \substack{ 1-A \\ 0 } \right.\right] \\
& =  \frac{1}{\Gamma(A)}H\substack{1,1\\1,1} \left[ x^{q}  \left| \substack{ (1-A,1) \\ (0,1) } \right.\right]
\end{split}
\end{equation}
and by employing \cite[Eq. (3.38.1)]{B:Prudnikov5} and \cite[Eq. (2.21), Eq. (1.58), Eq. (1.59)]{B:HFox}, $C_{q}(u)$ can be deduced as \eqref{Eq:aux}
yielding \eqref{Eq:ERMGFbased}, which completes the proof.
\end{proof}
It is noted that for the special case of uncorrelated $\mathcal{R}_\ell$, the EC of diversity receivers can be readily evaluated by employing the
following corollary.
\begin{corollary}
The EC of diversity receivers with $L$ independent diversity branches is deduced as
\begin{equation}\label{Eq:uncorrelated}
\begin{split}
R(\theta) = -\frac{1}{\theta T B}\ln{\left[\frac{1}{\Gamma(A)}\int_0^{\infty} C_{q}(u)\prod_{\ell=1}^L\mathcal{M}_{\mathcal{R}_\ell^p}(K_{p,q}u)\mathrm{d}u \right]}
\end{split}
\end{equation}
\end{corollary}
\begin{proof}
When independent branches are considered, $\mathcal{M}_{\vec{{\mathbf{R}}}^p}(K_{p,q}u)$ can be expressed as the product of the MGFs of $\mathcal{R}^p$, i.e.
$\mathcal{M}_{\vec{{\mathbf{R}}}^p}(K_{p,q}u) = \prod_{\ell=1}^L \mathcal{M}_{\mathcal{R}_\ell^p}(K_{p,q}u)$. Then, \eqref{Eq:uncorrelated} is readily obtained from \eqref{Eq:ERMGFbased}.
\end{proof}
Using Proposition~\ref{Prop:Proposition1}, for the special cases of MRC and EGC diversity receivers, i.e. for $q = 1$ and $2$, respectively, it can be shown that $C_{q}(u)$ can be expressed in terms of exponential and Bessel functions.
Specifically, for $q = 1$, by employing \cite[Eq. (8.4.3.1)]{B:Prudnikov3}, \cite[Eq. (8.2.2.8)]{B:Prudnikov3} and \cite[Eq. (8.2.2.16)]{B:Prudnikov3}, $C_{q}(u)$ can be deduced as
\begin{equation}\label{Eq:auxMRC}
C_{q}^{\mathrm{MRC}}(u) = u^{A-1}\exp(-u),
\end{equation}
which is in perfect agreement with \cite[Eq. (8)]{J:You}, \cite[Eq. (5)]{J:PeppasEC} and \cite[Eq. (7)]{J:Ji1}.
For $q = 2$, by employing \cite[Eq. (8.4.19.1)]{B:Prudnikov3}, \cite[Eq. (8.2.2.8)]{B:Prudnikov3} and \cite[Eq. (8.2.2.16)]{B:Prudnikov3}, $C_{q}(u)$ can be deduced as
\begin{equation}\label{Eq:auxEGC}
C_{q}^{\mathrm{EGC}}(u) = \sqrt{\pi}\left(\frac{u}{2}\right)^{A-1/2}J_{A-1/2}(u).
\end{equation}

Note that \eqref{Eq:auxMRC} has been derived in \cite{J:Ji1} and \cite{J:You} for the EC performance of MISO systems with transmit MRC over arbitrary generalized fading channels by using
\cite[Eq. (1.512/4)]{B:Gradshteyn_00} and the definition of the MGF.
However, our approach is more general as it includes this result as a special case\footnote{We became aware of \cite{J:Ji1} after we have completed the research reported in this paper. It is also noted that \cite{J:You} does not give reference to \cite{J:Ji1}.}.
Moreover, it is important to underline the generality of \eqref{Eq:auxMRC}, as it is also valid for other $L$-branch diversity structures. For example, in addition to the $L_p$-norm diversity schemes, it can be used to obtain the performance of selection diversity (SD) and generalized selection combining (GSC). In fact, the proposed analysis is valid as long as the MGF of $\gamma_{\rm end}$ is readily available.
For example, on the one hand, for the SD, $\gamma_{\rm end}$ can be expressed as $\gamma_{\rm end} = \max\{\mathcal{R}^2_1, \mathcal{R}^2_2, \ldots, \mathcal{R}^2_L\} $.
On the other hand,  for the GSC diversity scheme, the $K$ strongest branches out of $L$ available are combined as $\gamma_{\rm end} = \sum_{k=1}^K \mathcal{R}^2_{(k)}$ where $\mathcal{R}^2_{(1)} > \mathcal{R}^2_{(2)} > \ldots > \mathcal{R}^2_{(L)}$. For these two diversity schemes, EC can be computed by employing \eqref{Eq:ERMGFbased} with $C_{q}(u)$ given by \eqref{Eq:auxMRC}.
However, since analyzing the performance of such schemes is beyond the main scope of this paper, we will be focusing for the rest of the paper on the performance analysis of EGC and MRC diversity schemes.


\subsection{Computational Issues}
In the case of MRC, the numerical evaluation of EC is rather straightforward because $C_{q}^{\mathrm{MRC}}(u)$ is a smooth and well-behaving function. Specifically, EC can be evaluated numerically by employing a $N$-point Gauss-Chebyshev quadrature technique. Alternatively, standard built-in functions for numerical integration available in the most popular mathematical software packages, such as Matlab, Maple or Mathematica can be also used. For obtaining EGC, however, the numerical evaluation of a Hankel transform is required, rendering the evaluation considerably more difficult, since $C_{q}^{\mathrm{EGC}}(u)$ is oscillatory. However, efficient methods for the numerical evaluation of the Hankel transform are available in \cite{J:CreeHankel}.
Thus, the integrals in \eqref{Eq:ERMGFbased} are first expressed as
\begin{align}\label{Eq:Hankel}
\int_0^{\infty} C_{q}^{\mathrm{EGC}}(u)\mathcal{M}_{\vec{{\mathbf{R}}}^p}(K_{p,q}u)\mathrm{d}u = \nonumber \\
\sum_{k = 0}^\infty\int_{u_k}^{u_{k+1}} C_{q}^{\mathrm{EGC}}(u)\mathcal{M}_{\vec{{\mathbf{R}}}^p}(K_{p,q}u)\mathrm{d}u
\end{align}
where $u_0 = 0$ and $u_k$ is the $k$-th zero of $J_{A-1/2}(u)$, for $k \geq 1$.
The series in \eqref{Eq:Hankel}, however, is alternating so that a large number of terms may be required to achieve a sufficient numerical accuracy. Fortunately, by
invoking a convergence acceleration algorithm this series can be transformed into another that converges faster. Such an algorithm is the so-called Epsilon algorithm \cite{J:Shanks}.
The algorithm generates a two-dimensional array
called the $\epsilon$-table with entries $\epsilon_r^{(k)}$ where
$r$ is the column index and $k$ the location down the column.
At the initialization phase, the first column is set to zero as
$\epsilon_{-1}^{(k)} = 0$, $\forall k > 0$ and the second column is set to the given
partial sums $s_k$ of \eqref{Eq:Hankel}, i.e.  $\epsilon_{0}^{(k)} = s_k$, $k = 0,1, \ldots N-1$, where $N$ is the number of truncation terms. The remaining
elements of the $\epsilon$-table are deduced as
\begin{equation}
\epsilon_{r+1}^{(k)} = \epsilon_{r-1}^{(k+1)}+\frac{1}{\epsilon_{r}^{(k+1)} - \epsilon_{r}^{(k)}}, \, r = 1, 2, \ldots \,\,.
\end{equation}
The even columns of the $\epsilon$-table contain increasingly accurate estimates of the infinite series.

\section{Closed-form expressions for the EC performance of dual diversity receivers}\label{Sec:Anal}

In this section, by employing the MGF-based approach presented in Proposition~\ref{Prop:Proposition1}, closed-form expressions for the EC performance of dual, i.e. $L = 2$, diversity receivers are presented.
Hereafter, it is assumed that the receivers operate in the presence of GSNM fading channels
and AWGN. Thus, the unified MGF of the GSNM envelope distribution can be expressed as \cite[Eq. (24)]{J:YilmazEGC}
\begin{equation}\label{Eq:unifiedMGF}
\begin{split}
\mathcal{M}_{{\mathcal{R}_\ell^p}}(u) &= \frac{2}{\Gamma(m_{s\ell})\Gamma(m_{\ell})} \\
& \times H\substack{2,1\\1,2} \left[ \left(\frac{b_{\ell}m_{s\ell}}{\Omega_{s\ell}}\right)^p \frac{1}{u^2} \left| \substack{ (1,2) \\ (m_{s\ell},p),\, (m_{\ell},2p/\beta_{\ell}) } \right.\right]
\end{split}
\end{equation}
where $m_\ell$ $(0.5 \leq m_\ell < \infty)$ and $\beta_\ell$ $(0 \leq \beta_\ell < \infty)$ are the fading figure and the shaping factor, respectively, $m_{s\ell}$ $(0.5 \leq m_{s\ell} < \infty)$ and $\Omega_{s\ell}$ $(0 \leq \Omega_{s\ell} < \infty)$ denote the severity and the average power of shadowing, respectively, and $b_{\ell} = \Gamma(m_\ell+2/\beta_\ell)/\Gamma(m_\ell)$.

In the following proposition, a generic expression for the EC of $L_p$-norm receivers operating over GSNM fading channels is presented.
\begin{proposition}\label{Prop:DualDiversity}
The EC of $L_p$-norm diversity receivers operating over GSNM fading channels is given as \eqref{Eq:generic} (on the top of next page).
\begin{figure*}[!t]
\setcounter{mytempeqncnt}{\value{equation}}
\setcounter{equation}{14}
\begin{equation}\label{Eq:generic}
\begin{split}
R(\theta) &=
-\frac{1}{\theta T_f B}\ln\left\{
\frac{(2\pi)^{q/2-1/2}}{\sqrt{q}\Gamma(m_{s1})\Gamma(m_{s2})\Gamma(m_1)\Gamma(m_2)\Gamma(A)\pi}
 \right. H \, \substack{ 0 , 1 : 2 , 2 : 2 , 2 \\ q , 0 :2, 2 : 2 , 2} \left[
\left(\frac{\Omega_{s1}}{b_1m_{s1}}\right)^p\frac{K_{p,q}^2q^4}{4}
,
\left(\frac{\Omega_{s2}}{b_2m_{s2}}\right)^p\frac{K_{p,q}^2q^4}{4}\right.\\
& \left.
\left.
\left| \substack{ \left(1-A;\frac{2}{q}, \frac{2}{q}\right) \left(\frac{1}{q};\frac{2}{q}, \frac{2}{q}\right), \left(\frac{2}{q};\frac{2}{q}, \frac{2}{q}\right), \ldots, \left(\frac{q-1}{q};\frac{2}{q}, \frac{2}{q}\right) \\ - }  \right.
\left| \substack{ (1-m_{s1},p), \left(1-m_{1},\frac{2p}{\beta_1}\right) \\ \left(\frac{1}{2}, 1\right), (1,1)}  \right.
\left| \substack{ (1-m_{s1},p), \left(1-m_{1},\frac{2p}{\beta_1}\right) \\ \left(\frac{1}{2}, 1\right), (1,1)}  \right.
 \right]
 \right\}
\end{split}
\end{equation}
\hrulefill 
\setcounter{equation}{15}
\end{figure*}
\end{proposition}
\begin{proof}
From the definition of the Fox's H-function and by using the duplication formula for the gamma function \cite[Eq. (8.335/1)]{B:Gradshteyn_00}, the unified MGF in \eqref{Eq:unifiedMGF} can be expressed in terms of a Mellin-Barnes contour integral as
\begin{equation}\label{Eq:mellinfrac}
\begin{split}
\mathcal{M}_{{\mathcal{R}_\ell^p}}(u) &=
\frac{C_{\ell}}{2\pi\imath}\int_{\mathcal{C}_{\ell}}\Gamma(-t)\Gamma\left(\frac{1}{2}-t\right)
 \\
& \times
\Gamma(m_{s\ell}+p t)\Gamma\left(m_{\ell}+\frac{2p}{\beta_{\ell}}t\right)z_{\ell}^{-t}u^{2t}
\mathrm{d}t
\end{split}
\end{equation}
where $C_\ell = [{\sqrt{\pi}\Gamma(m_{s\ell})\Gamma(m_{\ell})}]^{-1}$ and $z_{\ell} = 4(b_{\ell}m_{s\ell}/\Omega_{s\ell})^p$
and $\mathrm{Re}\{t\} < -\min\{ m_{s\ell}/p, m_\ell\beta_\ell/(2p) \}$.
By substituting \eqref{Eq:mellinfrac} into \eqref{Eq:ERMGFbased} the following integral should be evaluated
\begin{equation}\label{Eq:int1}
\begin{split}
\mathcal{I} & = \left[\prod_{\ell =1}^2 C_{\ell}\right] \frac{1}{(2\pi\imath)^2}
\int_0^{\infty}\int_{\mathcal{C}_1}\int_{\mathcal{C}_2}
\prod_{\ell=1}^2\left[\Gamma(-t_{\ell})\Gamma\left(\frac{1}{2}-t_{\ell}\right) \right. \\
& \times \left.
\Gamma(m_{s\ell}+p t_{\ell})\Gamma\left(m_{\ell}+\frac{2p}{\beta_{\ell}}t_{\ell}\right)z_{\ell}^{-t_{\ell}}(K_{p,q}u)^{2t_{\ell}}
C_{q}(u)
\right] \\
& \times \mathrm{d}t_1\mathrm{d}t_2\mathrm{d}u.
\end{split}
\end{equation}
By employing the H-function representation of the auxiliary function $C_{q}(u)$, as well as \cite[Eq. (2.8)]{B:HFox}, the integral with respect to $u$
can be evaluated in closed form as
\begin{equation}
\begin{split}
&\int_0^{\infty}C_{q}(u)u^{2t_1+2t_2}\mathrm{d}u  = \\
& \frac{\Gamma(A+2t_1/q+2t_2/q)\Gamma(-2t_1/q-2t_2/q)}{\Gamma(-2t_1-2t_2)}.
\end{split}
\end{equation}
where $\mathrm{Re}\{A+2t_1/q+2t_2/q\} > 0$ and $\mathrm{Re}\{-2t_1/q-2t_2/q\} > 0$.
By further employing the multiplication formula of the Gamma function \cite[Eq. (8.335/1)]{B:Gradshteyn_00}, $\mathcal{I}$ can be expressed as the following two-fold Mellin-Barnes contour integral
\begin{equation}
\mathcal{I} = -\frac{\mathcal{D}}{4\pi^2}\int_{\mathcal{C}_1}\int_{\mathcal{C}_2}\phi(t_1, t_2)\phi_1(t_1)\phi_2(t_2)x_1^{t_1}x_2^{t_2}\mathrm{d}t_1\mathrm{d}t_2
\end{equation}
where
\begin{subequations}
\begin{equation}
\mathcal{D} = \frac{(2\pi)^{q/2-1/2}}{\sqrt{q}\Gamma(m_{s1})\Gamma(m_{s2})\Gamma(m_1)\Gamma(m_2)\Gamma(A)\pi},
\end{equation}
\begin{equation}
\phi(t_1, t_2) = \frac{\Gamma(A+2t_1/q+2t_2/q)}{\prod_{i=2}^q\Gamma\left(\frac{i-1}{q}-2t_1/q-2t_2/q\right)},
\end{equation}
\begin{equation}
\phi_\ell(t_\ell) = \Gamma(m_{s\ell}+pt_{\ell})\Gamma\left(m_{\ell}+\frac{2pt_{\ell}}{\beta_{\ell}}\right)\Gamma(0.5-t_{\ell})\Gamma(-t_{\ell}),
\end{equation}
\begin{equation}
x_\ell = \left(\frac{\Omega_{s_{\ell}}}{b_{\ell}m_{s_{\ell}}}\right)^p\frac{K_{p,q}^2q^4}{4}.
\end{equation}
\end{subequations}
Finally, from the definition of the bivariate H-function \cite[Eq. (2.56)-(2.60)]{B:HFox}, \eqref{Eq:generic} is readily obtained, thus completing the proof.
\end{proof}
Note that the bivariate H-function can be efficiently evaluated by using the Matlab algorithm presented in \cite{J:Peppas2012WCL}.
Based on Proposition~\ref{Prop:DualDiversity}, the EC of diversity receivers for several well known fading models, such as the generalized-K and the Nakagami-$m$, can be readily evaluated.
For example, let us consider dual-branch EGC diversity receivers over generalized-K fading channels. In this case, by selecting $\beta_{\ell} = 2$, $p = 1$, $q = 2$ in \eqref{Eq:generic} and by observing that the bivariate H-functions reduce to bivariate Meijer G-functions, EC is deduced as in \eqref{Eq:KG} (on the top of the next page).
\begin{figure*}[!t]
\setcounter{mytempeqncnt}{\value{equation}}
\setcounter{equation}{20}
\begin{equation}\label{Eq:KG}
\begin{split}
R(\theta) &=
-\frac{1}{\theta T_f B}\ln\left\{
\frac{1}{\sqrt{\pi}\Gamma(m_{s1})\Gamma(m_{s2})\Gamma(m_1)\Gamma(m_2)\Gamma(A)}
 \right. G \, \substack{ 0 , 1 : 2 , 2 : 2 , 2 \\ 2 , 0 :2, 2 : 2 , 2} \left[
\frac{\Omega_{s1}}{m_1m_{s1}}\frac{E_s}{2N_0}
,
\frac{\Omega_{s2}}{m_2m_{s2}}\frac{E_s}{2N_0} \right.\\
& \left.
\left.
\left| \substack{ 1-A; \frac{1}{2}, \\ - }  \right.
\left| \substack{ 1-m_{s1}, 1-m_{1} \\ \frac{1}{2}, 1}  \right.
\left| \substack{ 1-m_{s2}, 1-m_{2} \\ \frac{1}{2}, 1}  \right.
 \right]
 \right\}
\end{split}
\end{equation}
\hrulefill 
\setcounter{equation}{21}
\end{figure*}
\begin{figure*}[!t]
\setcounter{mytempeqncnt}{\value{equation}}
\begin{equation}\label{Eq:Nakagami}
\begin{split}
R(\theta) &=
-\frac{1}{\theta T_f B}\ln\left\{
\frac{1}{\sqrt{\pi}\Gamma(m_1)\Gamma(m_2)\Gamma(A)}
 \right. G \, \substack{ 0 , 1 : 2 , 1 : 2 , 1 \\ 2 , 0 :1, 2 : 1 , 2} \left[
\frac{\Omega_{s1}}{m_1}\frac{E_s}{2N_0}
,
\frac{\Omega_{s2}}{m_2}\frac{E_s}{2N_0}
\left.
\left| \substack{ 1-A; \frac{1}{2}, \\ - }  \right.
\left| \substack{  1-m_{1} \\ \frac{1}{2}, 1}  \right.
\left| \substack{  1-m_{2} \\ \frac{1}{2}, 1}  \right.
 \right]
 \right\}
\end{split}
\end{equation}
\hrulefill 
\setcounter{equation}{22}
\end{figure*}
For dual EGC diversity receivers operating over Nakagami-$m$ fading channels, by setting $m_{s\ell} \rightarrow \infty$ and employing the identity
$\lim_{x \rightarrow \infty} x^{-u}\Gamma(x+u)/\Gamma(x) = 1$ \cite[Eq. (8.328)]{B:Gradshteyn_00} along with the Mellin-Barnes integral representation of \eqref{Eq:generic}, EC is deduced as
\eqref{Eq:Nakagami} (on the top of next page).
Note that a Mathematica implementation of the bivariate G-function is available in \cite{J:Ansari}.

\section{Asymptotic Performance Analysis}\label{Sec:Asymptotic}
In this section, the asymptotic EC performance of MRC and EGC diversity receivers in the high- and low-SNR regimes is assessed.
More specifically, in the following proposition, a novel MGF-based approach for the evaluation of the EC in the high-SNR regime is first presented.
As will be seen later on, the proposed approach provides useful insights as to the factors
affecting EC performance in terms of the so-called high-SNR slope and high-SNR power offset \cite{J:ShamaiVerdu}.
\begin{proposition}\label{Prop:highSNR}
Let us assume that $\mathcal{M}_{\vec{{\mathbf{R}}}^p}(u)$ can be approximated for
$u\rightarrow \infty$ as
\begin{equation}\label{Eq:Taylor}
\mathcal{M}_{\vec{{\mathbf{R}}}^p}(u) = \mathcal{C} u^{-d} + o(u^{-d}).
\end{equation}
Then, for high SNR values, the EC of MRC and EGC diversity receivers can be approximated as
\begin{equation}\label{Eq:MRChigh}
\begin{split}
R^{\rm MRC}(\theta) & = \frac{d}{\theta T B}\left[\ln(E_s/N_0)
 -\frac{1}{d} \ln{\frac{\mathcal{C} \Gamma(A - d)}{\Gamma(A)}} \right],\,A > d
\end{split}
\end{equation}
and
\begin{equation}\label{Eq:EGChigh}
\begin{split}
R^{\rm EGC}(\theta) & = \frac{d}{2\theta T B}\left[\ln(E_s/N_0) -\ln{L} \right. \\
& \left.-\frac{2}{d} \ln{\frac{2^{-d}\sqrt{\pi}\mathcal{C} \Gamma(A - d/2)}{\Gamma(A)\Gamma\left(0.5+d/2\right)}} \right],\,d/2<A < d+1,
\end{split}
\end{equation}
respectively.
\end{proposition}
\begin{proof}
Following \cite{J:WangYannakis}, for high SNR values, $\mathcal{M}_{\vec{{\mathbf{R}}}^p}(u)$ can be approximated for
$u\rightarrow \infty$ as in \eqref{Eq:Taylor}.
Then, by substituting \eqref{Eq:Taylor} into \eqref{Eq:ERMGFbased}, and employing \eqref{Eq:auxMRC} and \eqref{Eq:auxEGC} along with \cite[Eq. (6.561/14)]{B:Gradshteyn_00} and the definition of the gamma function, asymptotic expressions for the EC of MRC and EGC diversity receivers, respectively, can be deduced as \eqref{Eq:MRChigh} and \eqref{Eq:EGChigh}, respectively.
\end{proof}
As it can be observed, both asymptotic expressions are of the form
\begin{equation}
R(\theta) = \mathcal{S}_{\infty}(\log_2(E_s/N_0) - \mathcal{L}_{\infty})
\end{equation}
where $\mathcal{S}_{\infty}$ is the high-SNR slope, defined as \cite{J:Verdu}
\begin{equation}
\mathcal{S}_{\infty} \triangleq \lim_{E_s/N_0 = 0}\frac{R(\theta)}{E_s/N_0}
\end{equation}
and $\mathcal{L}_{\infty}$ the high-SNR power offset defined as \cite{J:Verdu}
\begin{equation}
\mathcal{L}_{\infty} \triangleq \lim_{E_s/N_0 = 0}\left(\log_2(E_s/N_0)-\frac{R(\theta)}{\mathcal{S}_{\infty}} \right)
\end{equation}
Using \eqref{Eq:Taylor}, $\mathcal{S}_{\infty}$ and $\mathcal{L}_{\infty}$ can be obtained in a straightforward manner for any fading model, 
provided that the MGF $\mathcal{M}_{\vec{{\mathbf{R}}}^p}(u)$ is readily available.

Next, the low-SNR performance of the considered system is assessed in terms of
the minimum normalized energy per information bit to reliably convey any positive rate and the wide-band slope \cite{J:Verdu}.
\begin{proposition}\label{Prop:lowSNR}
The low-SNR EC performance of diversity receivers can be assessed by using a second-order expansion of the EC around $E_s/N_0 \rightarrow 0^+$ as 
\begin{equation}\label{Eq:RLow}
R(\theta) = \dot{R}(0)\frac{E_s}{N_0}+\frac{\ddot{R}(0)}{2}\left(\frac{E_s}{N_0}\right)^2+o\left(\frac{E_s}{N_0}\right),
\end{equation}
where $\dot{R}(0)$ and $\ddot{R}(0)$ denote the first and second order
derivatives of the EC with respect to the SNR, ${E_s}/{N_0}$, and are given by
\begin{equation}\label{Eq:dots}
\dot{R}(0) = \frac{1}{\sqrt{L^{1-p+q}}\ln 2}(-1)^p \left.\frac{\partial^p \mathcal{M}_{\vec{{\mathbf{R}}}^p}(u)}{\partial u^p}\right|_{u = 0}
\end{equation}
\begin{equation}\label{Eq:dots2}
\begin{split}
\ddot{R}(0) &= -\frac{A+1}{{L^{1-p+q}}\ln 2}\left.\frac{\partial^{2p} \mathcal{M}_{\vec{{\mathbf{R}}}^{2p}}(u)}{\partial u^{2p}}\right|_{u = 0} \\
 & +\frac{A}{{L^{1-p+q}}\ln 2}\left(\left.\frac{\partial^p \mathcal{M}_{\vec{{\mathbf{R}}}^p}(u)}{\partial u^p}\right|_{u = 0}\right)^2.
\end{split}
\end{equation}
\end{proposition}
\begin{proof}
Using \cite[Eq. (22)]{J:Matthaiou}, $\dot{R}(0)$ and $\ddot{R}(0)$ can be respectively expressed as
\begin{equation}\label{Eq:dots3}
\dot{R}(0) = \frac{1}{\sqrt{L^{1-p+q}}\ln 2}\mathbb{E}\langle \vec{{\mathbf{R}}}^p \rangle
\end{equation}
\begin{equation}\label{Eq:dots4}
\ddot{R}(0) = -\frac{A+1}{{L^{1-p+q}}\ln 2}\mathbb{E}\langle \vec{{\mathbf{R}}}^{2p} \rangle+\frac{A}{{L^{1-p+q}}\ln 2}\left(\mathbb{E}\langle \vec{{\mathbf{R}}}^p \rangle\right)^2.
\end{equation}
The expectations in \eqref{Eq:dots3} and \eqref{Eq:dots4} can be readily expressed in terms of $\mathcal{M}_{\vec{{\mathbf{R}}}^p}(u)$ as \cite[Eq. (5.67)]{B:Papoulis_02}
\begin{equation}\label{Eq:momsMGF}
\begin{split}
& \mathbb{E}\langle \vec{{\mathbf{R}}}^{n p} \rangle =  (-1)^n \left.\frac{\partial^n \mathcal{M}_{\vec{{\mathbf{R}}}^p}(u)}{\partial u^n}\right|_{u = 0}.
\end{split}
\end{equation}
yielding \eqref{Eq:dots} and \eqref{Eq:dots2}, thus completing the proof.
\end{proof}
As was pointed out in \cite{J:Matthaiou},
$\dot{R}(0)$ and $\ddot{R}(0)$ are related two metrics, namely to the minimum normalized energy
per information bit to reliably convey any positive rate and
the wide-band slope, respectively \cite{J:Verdu}.
When QoS constraints are considered, these metrics can be defined respectively as \cite[Eq. (19)]{J:Matthaiou}
\begin{equation}
\frac{E_b}{N_0}_{\min} \triangleq \frac{1}{\dot{R}(0)},\, S \triangleq -\frac{2\ln 2 \left[\dot{R}(0)\right]^2 }{\ddot{R}(0)}.
\end{equation}
Finally, it is underlined that the above asymptotic performance analysis
in both the high- and low-SNR regimes is valid for all well-known fading distributions as long as MGF of $\vec{{\mathbf{R}}}^p$, $\mathcal{M}_{\vec{{\mathbf{R}}}^p}(u)$, exists.

\section{Performance evaluation results and discussion}\label{Sec:Results}
\parskip = 0pt
In this section, the proposed analytical approach presented in the previous sections is employed to obtain the EC of diversity receivers.
More specifically, the following case studies are considered:
\emph{a)} Dual MRC and EGC diversity receivers operating over GSNM fading channels;
\emph{b)} MRC diversity receivers operating over arbitrary correlated generalized Rice fading channels; \emph{c)}
EGC diversity receivers over (i.i.d) generalized gamma (GG) fading channels; \emph{d)} $\alpha-\kappa-\mu$ and \emph{e)} $\alpha-\eta-\mu$ fading channels.
For the last three cases, in addition to the exact analysis, asymptotic performance evaluation results are also included.

\subsection{
Dual MRC and EGC over Uncorrelated GSNM Fading Channels}
Hereafter, we evaluate the EC performance of dual MRC and EGC diversity receivers over GSNM fading channels with
$(m_1, m_2) = (1.25, 2.25)$, $(m_{s1}, m_{s2}) = (2.3, 1.3)$, $(\beta_1, \beta_2) = (5/3, 4/3)$,
and $(\Omega_{s1}, \Omega_{s1}) = (3.5, 2.5)$. The limiting cases of generalized-K fading, i.e. $(\beta_1, \beta_2) = (2, 2)$, and
Nakagami-$m$ fading, i.e. $(\beta_1, \beta_2) = (2, 2)$ and  $(m_{s1}, m_{s2}) \rightarrow (\infty, \infty)$, are further considered.
Fig.~\ref{Fig:RateGSNM} depicts the EC of all four case studies as a function of $E_s/N_0$, assuming $A = 4$. It can be seen from this figure that analytical and computer simulation results match very well, thus verifying the correctness of our analysis.
\begin{figure}[!t]
\centering
\includegraphics[keepaspectratio,width=3.5in]{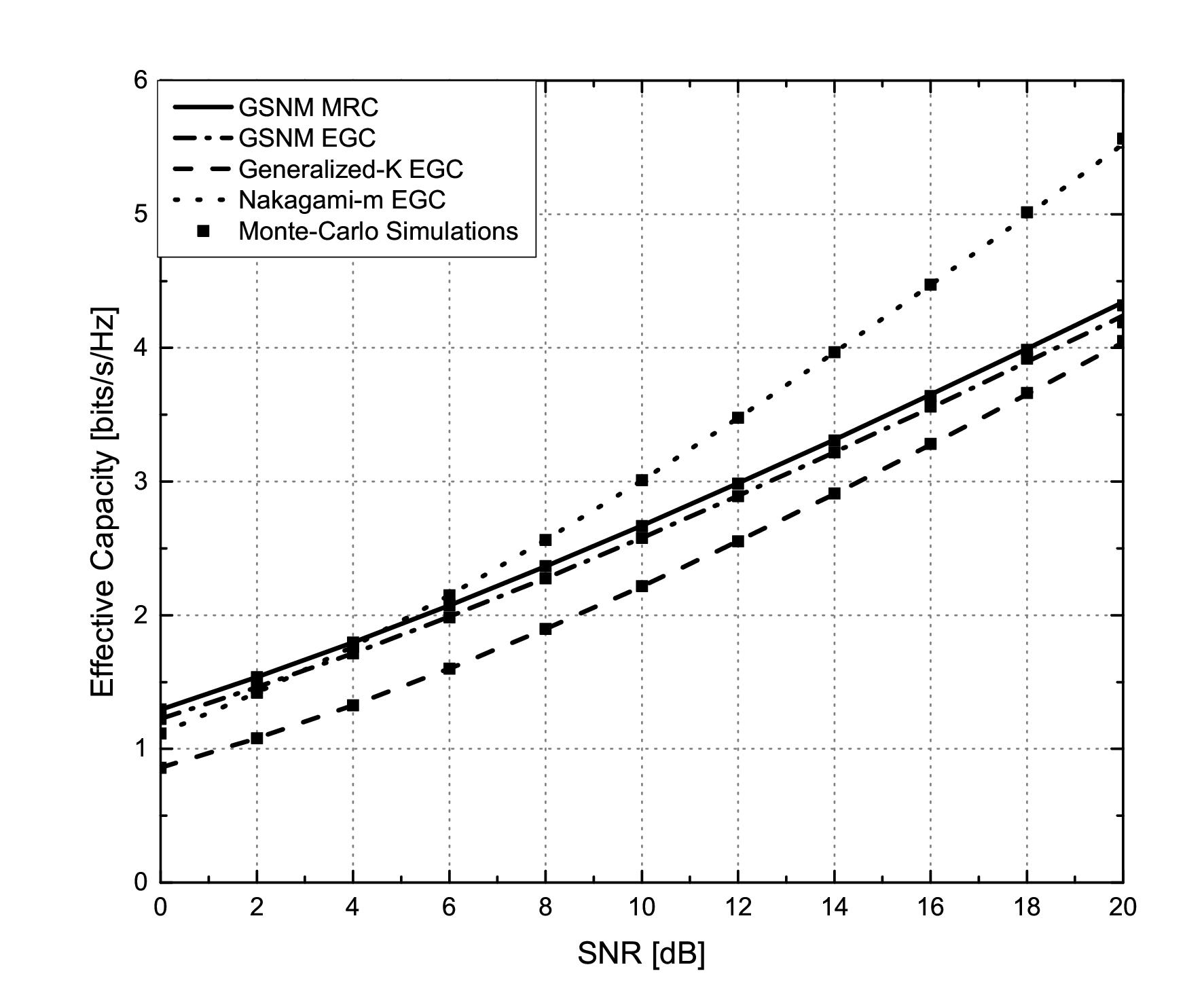}
\caption{Exact analytical EC performance evaluation results for dual MRC and EGC diversity receivers in GSNM fading}\label{Fig:RateGSNM}
\end{figure}

\subsection{
MRC over Arbitrarily Correlated Generalized Rice Fading Channels}
\parskip = 0pt
Considering generalized Rice fading channels with arbitrary correlation, the instantaneous SNR at the output of the $L$-branch receiver is given as $\gamma_{\rm end} = E_s/N_0\sum_{\ell=1}^L \mathcal{R}_\ell^2 = \mathbf{h} \mathbf{h}^H E_s/N_0$  where $\mathbf{h}$ is a $L \times 1$ complex Gaussian random vector,
having mean value $\boldsymbol{\eta}$ and covariance matrix $\mathbf{C} = \mathbb{E} \langle(\mathbf{h} -
\boldsymbol{\eta}) \,
(\mathbf{h} -
\boldsymbol{\eta})^H \rangle$.
\subsubsection{Exact Analysis}
\parskip = 0pt
By employing \cite[Eq. (18)]{J:YYilmaz}, the MGF of $\parallel\mathbf{h}\parallel^2_F$ can be deduced as
\begin{align}\label{Eq:MGFcorr}
\mathcal{M}_{\parallel\mathbf{h}\parallel^2_F}(u) = \frac{\exp\left[-u\,{\mu}\,\boldsymbol{\eta}^H\left(\mathbf{I}+u\mathbf{C}\right)^{-1}\boldsymbol{\eta}\right]}{\left[\det\left(\mathbf{I}+u\mathbf{C}\right)\right]^{\mu}}
\end{align}
where the parameter ${\mu}$ denotes the diversity order of the signal at each branch,
Note that \eqref{Eq:MGFcorr} is quite versatile, as it includes as special cases
the MGF of $\parallel\mathbf{h}\parallel^2_F$ over correlated Nakagami-$m$ ($\boldsymbol{\eta} = 0$)
\cite[Eq. (9.219)]{B:Alouini} and correlated Rice (${\mu} = 1$) fading channels \cite[Eq. (5)]{J:Hamdi}.

\subsubsection{Asymptotic Analysis}
\parskip = 0pt
For high SNR, i.e. for $s\rightarrow \infty$, the MGF of $\parallel\mathbf{h}\parallel^2_F$ in \eqref{Eq:MGFcorr} can be approximated as
\begin{align}\label{Eq:MGFcorrH}
\mathcal{M}^{\infty}_{\parallel\mathbf{h}\parallel^2_F}(u) \approx \frac{\exp\left[{\mu}\,\boldsymbol{\eta}^H\boldsymbol{\eta}\right]}{\left[\det\left(\mathbf{C}\right)\right]^{\mu}} u^{-\mu},
\end{align}
which is of the form given in \eqref{Eq:Taylor}, with $\mathcal{C} = \frac{\exp\left[{\mu}\,\boldsymbol{\eta}^H\boldsymbol{\eta}\right]}{\left[\det\left(\mathbf{C}\right)\right]^{\mu}}$ and $d = \mu$. Therefore, using \eqref{Eq:MRChigh}, the high-SNR asymptotic EC under generalized fading with MRC diversity receiver  can be easily obtained.

Using the above mentioned procedure, the exact analytical and asymptotic EC performance of the considered system with $L = 3$ branches,
has been obtained and is illustrated in Fig.~\ref{Fig:RateGCorrH} for various values of $A$.
These results have been obtained using $\boldsymbol{\eta} = \left[0.25\exp(\imath\pi/4)\, 0.5\exp(\imath\pi/6)\, \exp(\imath\pi/8)\right]^T$
and covariance matrix given by \cite{J:YYilmaz}:
\begin{equation}
\mathbf{C} =
\left( \begin{array}{ccc}
1 & 0.5\, e^{\frac{\imath\pi} 2} & 0.25\, e^{\frac{\imath\pi} 4}  \\
0.5\, e^{-\frac{\imath\pi} 2} & 2 & 0.125\, e^{\frac{\imath\pi} 8} \\
0.25\, e^{-\frac{\imath\pi} 4} & 0.125\, e^{-\frac{\imath\pi} 8} &  3
\end{array} \right) .
\end{equation}
As it can be observed, EC increases as $\mu$ increases or $A$ decreases.
In the same figure, equivalent Monte-Carlo computer simulated performance evaluation results have been obtained which perfectly match the analytical ones.
Moreover, it is evident that the high-SNR approximation provides exact results, even at medium SNR values, and can thus accurately predict the respective EC.
\begin{figure}[!t]
\centering
\includegraphics[keepaspectratio,width=3.5in]{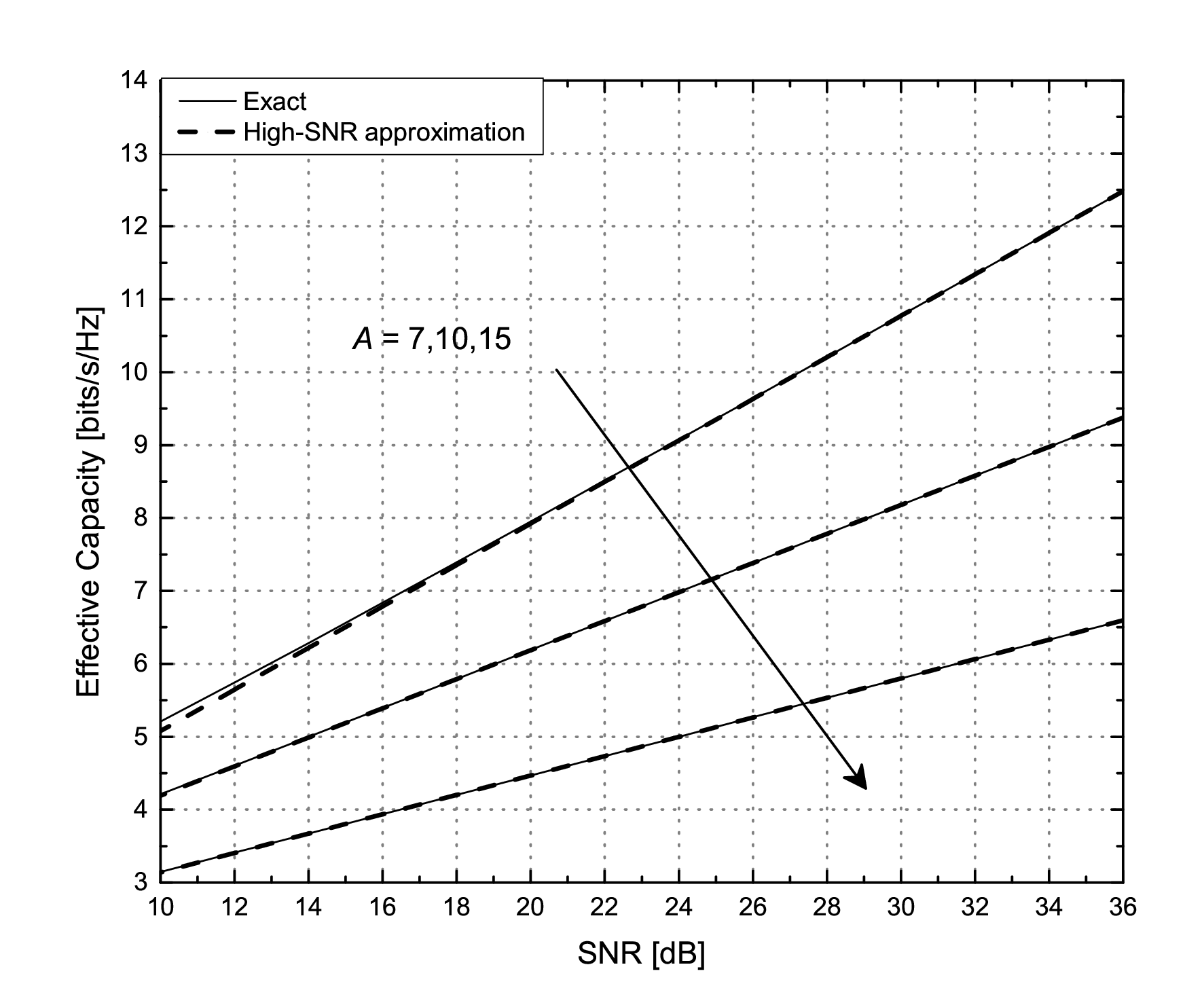}
\caption{Exact analytical EC performance evaluation results and their corresponding and high-SNR approximation for MRC diversity receivers with $L = 3$ branches in correlated generalized Rice fading ($\mu=2$)}\label{Fig:RateGCorrH}
\end{figure}

\subsection{
MRC or EGC over Uncorrelated Generalized Gamma Fading Channels }
\parskip = 0pt
In this case, the PDF or $\mathcal{R}_\ell$ is given as \cite{J:Stacy}
\begin{align}\label{Eq:PDFGG}
f_{\mathcal{R}_\ell}(r) = \frac{\beta \left({b}/{\Omega}\right)^{{m\beta}/{2}}r^{\beta m-1}\exp\left[{-\left(\frac{b}{\Omega}\right)^{{\beta}/{2}}r^{\beta}}\right] }{\Gamma(m)},
\end{align}
where $\beta >0$ and $m>1/2$ are two parameters related to the average fading severity and $\Omega$ is related to the average fading power as  $\mathbb{E}\langle \mathcal{R}{_\ell}{^2}\rangle=\Big({\Omega}/{m}\Big)^{2/\beta}{\Gamma(m+2/\beta)}/{\Gamma(m)}$.

\subsubsection{Exact Analysis}
\parskip = 0pt
For rational values of $\beta$, the MGF of $\mathcal{R}_\ell^p$ can be obtained in terms of Meijer's G-function by using a similar line of arguments as in \cite{J:Sagiasgg}. For arbitrary values of $\beta$, the MGF can be expressed in terms of the Fox's H-function by employing the approach presented in \cite{C:Yilmaz4}.
In what follows, a simple, computational efficient expression for the MGF of $\mathcal{R}_\ell^p$, valid for arbitrary values of $\beta$ will be derived.
By employing the definition of the MGF and performing the change of variables $\left({b}/{\Omega}\right)^{{\beta}/{2}}r^{\beta} = t^2$,
$\mathcal{M}_{{\mathcal{R}_\ell^p}}(u)$ can be deduced as
\begin{equation}\label{Eq:MGFGGaux}
\mathcal{M}_{{\mathcal{R}_\ell^p}}(u) = \frac{2}{\Gamma(m)}\int_0^{\infty} {t}^{2m-1}\exp\left[-t^2-u {t}^{\frac{2p}{\beta}\big({\Omega}/{b}\big)^{{p}/{2}}}\right]\mathrm{d}t
\end{equation}
The integral in \eqref{Eq:MGFGGaux} can be solved by employing a Gauss-Chebyshev quadrature technique as
\begin{equation}\label{Eq:MGFGG}
\mathcal{M}_{{\mathcal{R}_\ell^p}}(u) = \frac{2}{\Gamma(m)}\sum\limits_{k=1}^{N_t} w_k {t_k}^{2m-1}\exp\left[{-u {t_k}^{\frac{2p}{\beta}\left({\Omega}/{b}\right)^{{p}/{2}}}}\right],
\end{equation}
where $N_t$ is the number of integration points, $w_k$ and $t_k$ are the weights and abscissae given in \cite{J:Steen}.
%

\subsubsection{Asymptotic Analysis}
\parskip = 0pt
Hereafter we consider the EGC case only, i.e. $q = 2$.
Then by using \eqref{Eq:EGChigh}, the high-SNR asymptotic EC performance under GG fading with EGC diversity can be easily deduced.
Specifically, by employing a Taylor series expansion of \eqref{Eq:PDFGG} at $r = 0$ one obtains
\begin{equation}
f_{\mathcal{R}_\ell}(r) \approx \frac{\beta \left({b}/{\Omega}\right)^{{m\beta}/{2}}r^{\beta m-1}}{\Gamma(m)}.
\end{equation}
By invoking the definitions of the MGF and the gamma function, the MGF of $\mathcal{R}_\ell$ when $u\rightarrow \infty$ can be approximated as
\begin{align}\label{Eq:MGFGGH}
\mathcal{M}_{\mathcal{R}_\ell}(u) \approx \frac{\beta \big({b}/{\Omega}\big)^{{m\beta}/{2}}\Gamma(\beta m)}{\Gamma(m)}u^{-\beta m}.
\end{align}
In Fig.~\ref{Fig:RateGGEGCH}, the exact analytical EC performance and high-SNR approximation of EGC diversity receivers with $L$ = 3 branches over i.i.d. GG fading channels are depicted for various values of $A$. All results have been obtained by assuming $N_t$ = 15. It is evident that the high-SNR approximation provides very tight results for high SNR values.

\begin{figure}[!t]
\centering
\includegraphics[keepaspectratio,width=3.5in]{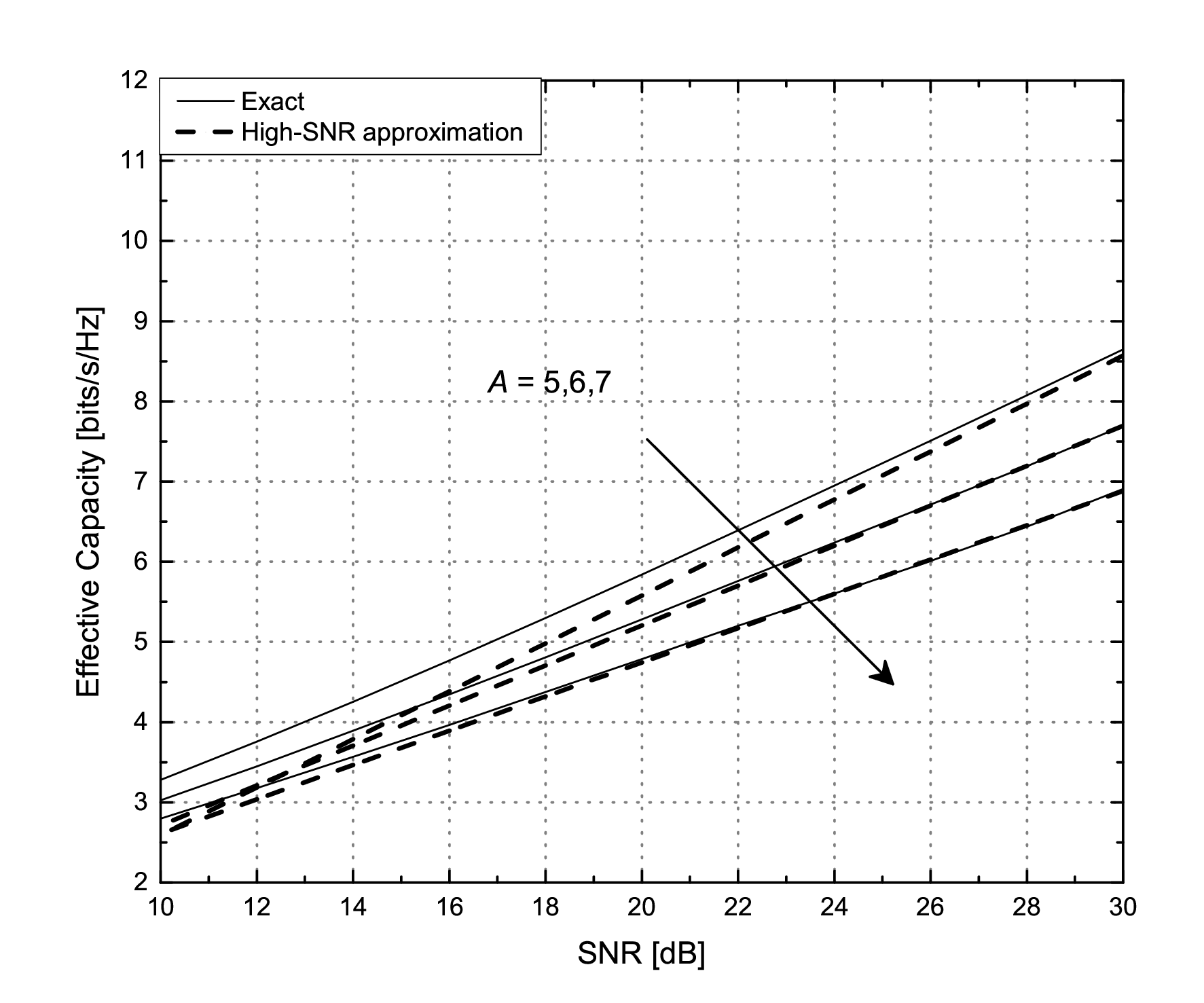}
\caption{Exact analytical EC performance evaluation results and their high-SNR approximation for EGC diversity receivers with $L = 3$ branches in i.i.d. GG fading ($m=2, \beta=1.5$) }\label{Fig:RateGGEGCH}
\end{figure}

It is noted that the low-SNR asymptotic EC performance under GG fading with EGC diversity can be readily deduced by combining \eqref{Eq:RLow} and \eqref{Eq:dots}\textendash\eqref{Eq:momsMGF}.
In Fig.~\ref{Fig:RateGGEGCL}, the exact analytical EC and low-SNR approximation of EGC receivers with $L = 3$ branches and operating over i.i.d. GG fading channels are depicted for various values of $A$. It is evident from these results that the low-SNR approximation is very accurate.
\begin{figure}[!t]
\centering
\includegraphics[keepaspectratio,width=3.5in]{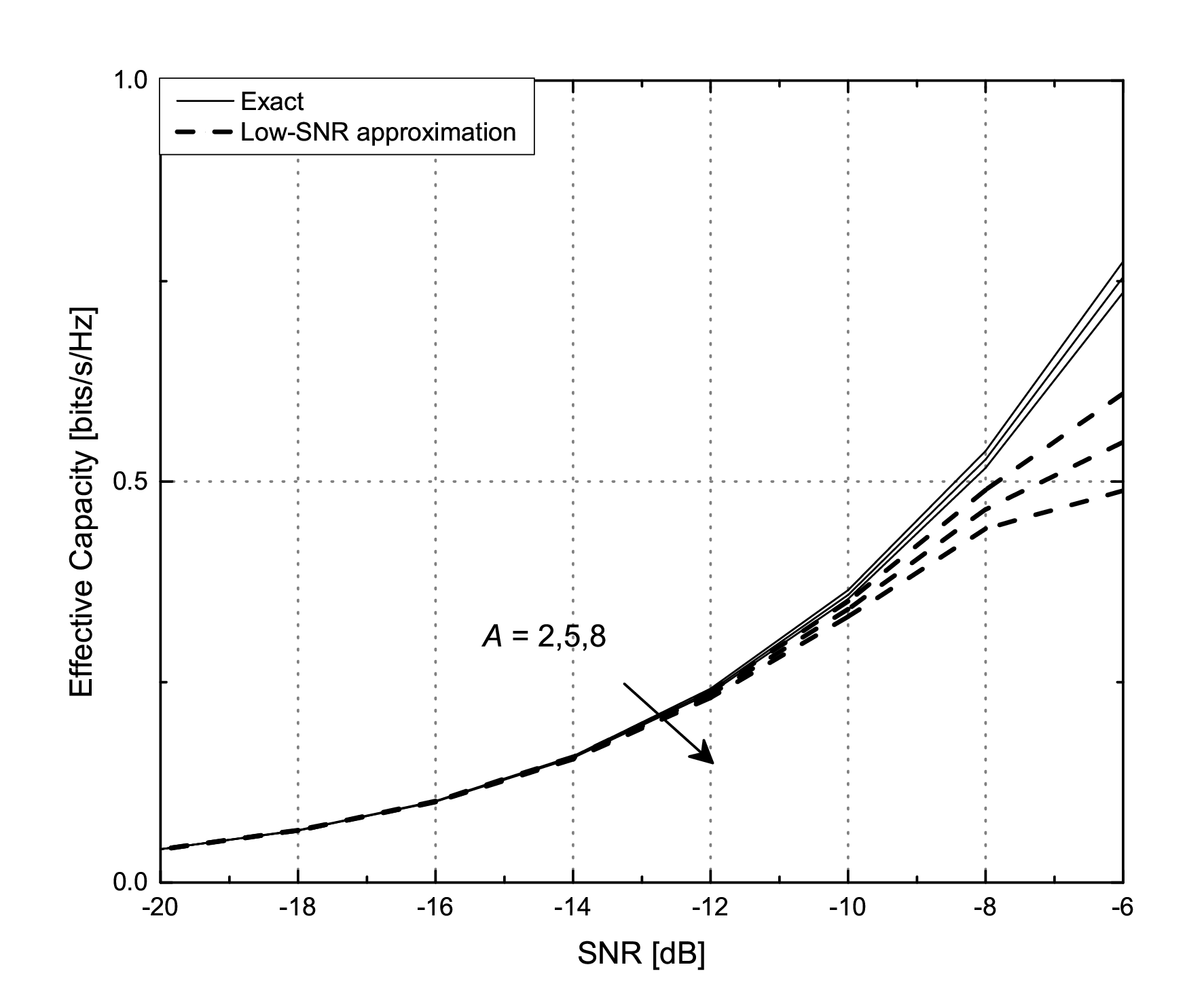}
\caption{Exact analytical EC performance evaluation results and their low-SNR approximation for EGC receivers with $L = 3$ branches in i.i.d. GG fading ($m=3, \beta=3$) }\label{Fig:RateGGEGCL}
\end{figure}


\subsection{MRC or EGC over Uncorrelated $\alpha-\kappa-\mu$ Fading Channels}
\parskip = 0pt
The $\alpha-\kappa-\mu$ is a very general fading model that includes as special cases several well-known distributions, namely the generalized gamma, the $\kappa-\mu$, the Nakagami-$m$ and the Rice distribution.
In this case, the PDF of $\mathcal{R}_\ell$ is given by \cite{C:Fraidenraich}
\begin{align}\label{Eq:PDFakm}
\begin{split}
f_{\mathcal{R}_\ell}(r)&=\frac{\alpha\kappa^{\frac{1-\mu}{2}}(1+\kappa)^{\frac{1+\mu}{2}}\mu r^{\frac{\alpha(1+\mu)}{2}-1}}{\exp\left[{\mu(\kappa+r^{\alpha}+\kappa r^{\alpha})}\right]}\\&\times I_{\mu-1}\left(2\sqrt{\kappa(1+\kappa)\mu r^{\alpha/2}}\right),
\end{split}
\end{align}
where $\alpha$, $\mu$ and $\kappa\geq0$ are the distribution parameters.
\subsubsection{Exact Analysis}
\parskip = 0pt
A novel computationally efficient expression for $\mathcal{M}_{{\mathcal{R}_\ell^p}}(u)$ can be obtained by following a similar line of arguments as in the generalized gamma case, as
\begin{align}\label{Eq:MGFakm}
\begin{split}
& \mathcal{M}_{{\mathcal{R}_\ell^p}}(u) = 2\exp[{-\mu\kappa}](\kappa\mu)^{(1-\mu)/{2}}\\&\times\sum \limits_{k=1}^{N_t}w_k{t_k}^{\mu}I_{\mu-1}(2\sqrt{\kappa\mu}t_k)\exp\left[-u\left(\frac{t_k^2}{\mu(1+\kappa)}\right)^{p/a}\right].
\end{split}
\end{align}
where the number of integration points $N_t$, the weights $w_k$ and abscissae $t_k$ are given in \cite{J:Steen}.
\begin{figure}[!t]
\centering
\includegraphics[keepaspectratio,width=3.5in]{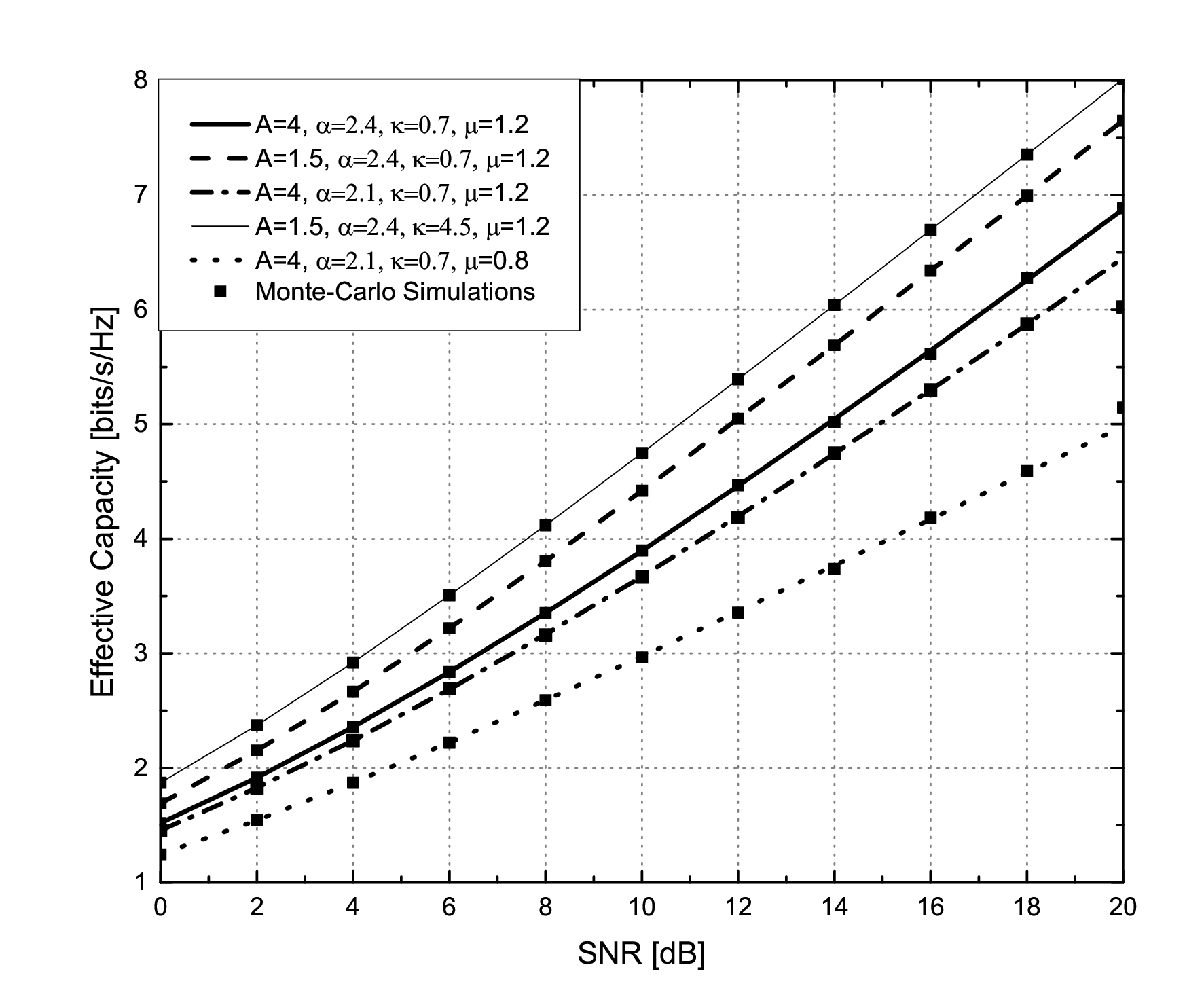}
\caption{Exact analytical EC performance evaluation results for EGC receivers with $L = 3$ branches in i.i.d. $\alpha-\kappa-\mu$ fading}\label{Fig:RateakmEGC}
\end{figure}

Fig.~\ref{Fig:RateakmEGC} depicts the EC performance of an EGC receiver with  $L = 3$ diversity branches, under i.i.d. $\alpha-\kappa-\mu$ fading.
As it can be observed, analytical and simulated results are again in perfect agreement. Moreover, as expected, EC performance improves as $A$ decreases and $\alpha$, $\kappa$ or $\mu$ increases.
\subsubsection{Asymptotic Analysis}
\parskip = 0pt
Hereafter, the EGC case is considered, i.e. $q = 2$.
By following a similar line of arguments as in the generalized gamma case, $\mathcal{M}_{{\mathcal{R}_\ell}}(u)$ when $u\rightarrow \infty$ can be approximated as
\begin{align}\label{Eq:MGFakmH}
\mathcal{M}_{{\mathcal{R}_\ell}}(u) \approx \alpha(1+\kappa)^{\mu}\mu^{\mu}\exp({-\mu\kappa})\Gamma(\alpha\mu) u^{-\alpha\mu}.
\end{align}
Therefore using \eqref{Eq:MRChigh}\textendash\eqref{Eq:EGChigh} asymptotically tight high-SNR expressions for the EC performance of the considered system can be readily deduced.
\subsection{
MRC or EGC over Uncorrelated $\alpha-\eta-\mu$ Fading Channels}
\parskip = 0pt
The $\alpha-\eta-\mu$ distribution is a very general fading distribution that includes as special cases the generalized gamma, the $\eta-\mu$, the Hoyt and the Nakagami-$m$ distributions \cite{C:Fraidenraich}.
In this case, the PDF of $\mathcal{R}_\ell$ is given by \cite{C:Fraidenraich}
\begin{align}\label{Eq:PDFahm}
\begin{split}
f_{{\mathcal{R}_\ell}}(r)&=\frac{\alpha(\eta-1)^{\frac{1}{2}-\mu}(\eta+1)^{\frac{1}{2}+\mu}r^{\alpha(\frac{1}{2}+\mu)-1}}{\exp\left[-{(1+\eta)^2\mu r^{\alpha}}/({2\eta})\right]\sqrt{\eta}\Gamma(\mu)}\\&\times\sqrt{\pi}\mu^{\frac{1}{2}+\mu}I_{\mu-\frac{1}{2}}\Big(\frac{(\eta^2-1)\mu r^{\alpha}}{2\eta}\Big),
\end{split}
\end{align}
where $\eta\geq0$ .
\subsubsection{Exact Analysis}
\parskip = 0pt
Following a similar methodology as in the previous test cases, a novel expression for the MGF of $\mathcal{R}_\ell^p$ is deduced as
\begin{align}\label{Eq:MGFahm}
\begin{split}
&\mathcal{M}_{{\mathcal{R}_\ell^p}}(u) = \frac{2\alpha\sqrt{\pi}(\eta-1)^{\frac{1}{2}-\mu}(\eta+1)^{\frac{1}{2}+\mu}\mu^{\frac{1}{2}+\mu}}{p\sqrt{\eta}\Gamma({\mu})}\\&\times\sum \limits_{k=1}^{N_t}w_k{t_k}^{{\alpha}(1+2\mu)/{p}-1}\exp\left[-\frac{\mu(1+\eta)^2t_k^{{2\alpha}/{p}}}{u^{{\alpha}/{p}}\eta}\right]\\&\times {I_{\mu-\frac{1}{2}}\left[{(\eta^2-1)\mu t_k^{{2\alpha}/{p}}}/({2u^{{\alpha}/{p}}\eta})\right]}{u^{-{\alpha(1+2\mu)}/{(2p)}}}.
\end{split}
\end{align}
where $N_t$ is the number of integration points, $w_k$ and $t_k$ are the weights and abscissae given in \cite{J:Steen}.
\subsubsection{Asymptotic Analysis}
\parskip = 0pt
Following a similar reasoning as in the previous cases and assuming EGC diversity reception, $\mathcal{M}_{{\mathcal{R}_\ell}}(u)$ can be approximated as
\begin{align}\label{Eq:MGFahmH}
\begin{split}
\mathcal{M}_{{\mathcal{R}_\ell}}(u) &\approx \frac{\alpha(\eta-1)^{\frac{1}{2}-\mu}(1+\eta)^{\frac{1}{2}+\mu}\mu^{\frac{1}{2}+\mu}\sqrt{\pi}\Gamma(2\alpha\mu)}{\sqrt{\eta}\Gamma(\mu)}\\&\times\Big[\frac{(\eta^2-1)\mu}{4\eta}\Big]^{\mu-\frac{1}{2}} u^{-2\alpha\mu}.
\end{split}
\end{align}
In Fig.~\ref{Fig:RateahmEGCH}, the exact EC and high-SNR approximation performance of EGC receivers with $L = 3$ branches and operating over i.i.d. $\alpha-\eta-\mu$ fading channels are depicted for various values of A. It is evident that the high-SNR approximation provides exact results, even at moderate SNR values, and can thus accurately predict the respective effective capacity.

\begin{figure}[!t]
\centering
\includegraphics[keepaspectratio,width=3.5in]{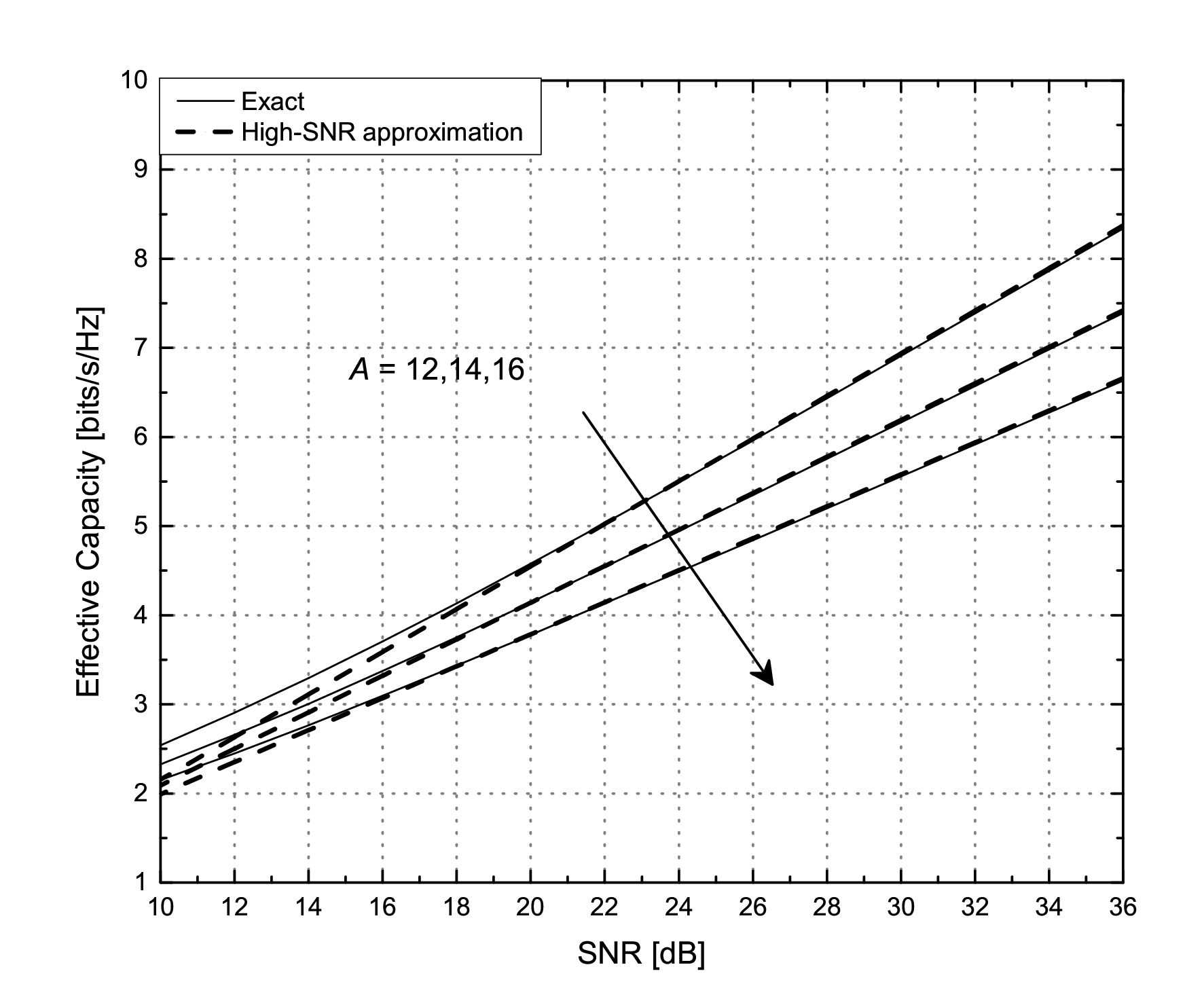}
\caption{Exact analytical EC performance evaluation results and their high-SNR approximation of EGC diversity receivers with $L = 3$ branches in i.i.d $\alpha-\eta-\mu$ fading ($\alpha=2.4, \eta=64.3, \mu=1.2$) }\label{Fig:RateahmEGCH}
\end{figure}

\section{Conclusions}\label{Sec:Conclusions}
Real-time applications are quite delay-sensitive, requiring an alternative performance metric rather than the conventional Shannon or outage capacity. Lately, the EC has attracted attention as a suitable metric quantifying end-to-end system performance under QoS limitations.
In this paper a new MGF-based methodology for obtaining,
in a unified way, the exact EC performance of $L_p$-norm diversity reception over arbitrary and correlated generalized fading channels, was proposed.
For the special case of dual diversity, closed-form expressions for the EC performance over GSNM fading channels have further been deduced.
Finally, an analytical MGF-based approach for the asymptotic analysis of the EC performance at low- and high-SNR regions was also proposed thus providing useful insights regarding the operating parameters which affect the overall system performance.
The validity of the proposed analytical methodology was assessed by considered very generic channel fading models that describe wireless propagation in a more realistic manner than the conventional fading models.
 The accuracy of the proposed analysis was substantiated with numerical results, accompanied with equivalent performance evaluation results obtained by means of Monte-Carlo simulations.


\bibliographystyle{IEEEtran}
\bibliography{IEEEabrv,peppas2}
%

\end{document}